\documentclass[journal,12pt,onecolumn,draftclsnofoot]{IEEEtran}
\renewcommand\IEEEkeywordsname{Index Terms}

\usepackage{amsmath,amsfonts,graphicx,tabularx}

\newtheorem{theorem}{Theorem}[section]
\newtheorem{lemma}[theorem]{Lemma}

\newtheorem{corollary}[theorem]{Corollary}
\newtheorem{definition}{Definition}[section]

\newtheorem{example}{Example}[section]
\newtheorem{remark}{Remark}[section]

\ifCLASSINFOpdf
\else
\fi
\begin{document}
\title{Analysis of Statistical Properties of Nonlinear Feedforward Generators Over Finite Fields}

\author{Suman Roy and Srinivasan Krishnaswamy% <-this % stops a space
\thanks{The authors are with the Department of Electronics and Electrical Engineering, Indian Institute
of Technology Guwahati, Guwahati-781039, Assam, India (e-mail: suman.roy@iitg.ernet.in; srinikris@iitg.ernet.in).}% <-this % stops a space
}

\maketitle

\begin{abstract}
Due to their simple construction, LFSRs are commonly used as building blocks in various random number generators. Nonlinear feedforward logic is incorporated in LFSRs to increase the linear complexity of the generated sequence. In this work, we extend the idea of nonlinear feedforward logic to LFSRs over arbitrary finite fields and analyze the statistical properties of the generated sequences. Further, we propose a method of applying nonlinear feedforward logic to word-based $\sigma$-LFSRs and show that the proposed scheme generates vector sequences that are statistically more balanced than those generated by an existing scheme.  
\end{abstract}


% \fi
\IEEEpeerreviewmaketitle

\section{Introduction}
Pseudorandom number generators (PRNGs) \cite{Paar2009} have a wide array of applications ranging from cryptography (\cite{Paar2009}, \cite{menezes1996}) and error correcting codes  \cite{peterson1972error} to spread spectrum communication \cite{1095533}. Due to their simple construction and ease of hardware implementation linear feedback shift registers (LFSRs) are commonly used as basic building blocks for PRNGs.  For a given number of delay blocks, LFSRs with primitive characteristic polynomials generate sequences with maximum period. Such sequences have a balanced distribution of 0's and 1's and exhibit properties like the span-$n$ property and $2$-level autocorrelation which are desirable for randomness \cite{Golomb:1981:SRS:578271}. However, sequences generated by LFSRs are marred by their low linear complexity. One way of increasing the linear complexity of such sequences is by the use of nonlinear feedforward logic \cite{Groth1971}. An analysis of the linear complexity of binary sequences generated by nonlinear feedforward generatetors (NLFGs) is given in \cite{key1976}. Statistical properties of such sequences are investigated in \cite{Dawson1990}, \cite{bedi2001}, \cite{gammel2006}, \cite{teo2013}. In this paper, we have analyzed sequences generated by NLFGs where the underlying LFSR implements a linear recurring relation (LRR) in an arbitrary finite field. Further, we have proposed a method of applying nonlinear feedforward logic to $\sigma$-LFSRs. We have then compared the statistical distribution of sequences generated by the proposed scheme with those generated by the scheme mentioned in \cite{sartaj2012}.  

The remainder of this paper is organized as follows. Section \ref{sec:Linear feedback shift registers} contains an introduction to LFSRs and motivates the use of NLFGs. Section \ref{sec3} describes NLFGs and analyzes the properties of sequences generated by them. Section \ref{sec 4} describes an implementation of NLFGs over word-based $\sigma$-LFSRs and contains a statistical analysis of sequences generated by such a configuration. Section \ref{sec5} briefly summarizes the paper.
 
The notations used in this paper are as follows. The cardinality of a set $S$ is denoted by $|S|$. $\mathbb{F}_q$ denote the finite field of order $q = p^{n}$, where $p$ is a prime number and $n$ is a positive integer. $\mathbb{F}_{q}^{n}$ denotes the $n$-dimensional vector space over $\mathbb{F}_{q}$.

\section{Linear Feedback Shift Registers}
\label{sec:Linear feedback shift registers}

An $L$-stage feedback shift register (FSR) is a circuit consisting  of $L$ delay blocks along with a feedback function $f$. An $L$-stage FSR generates  a sequence $\{s_i\}_{i=0}^{\infty} =\{s_0,s_1,s_2\ldots\}$ where elements are related by a recurrence relation $s_{j+L}=f(s_j,s_{j+1},\ldots, s_{j+L-1})$. If the function $f$ is linear then the FSR is called a linear feedback shift register (LFSR). Figure \ref{LFSR} depicts an LFSR having $L$ delay blocks with a linear feedback loop.

\begin{figure}[h]
\centering
\includegraphics[scale=0.7]{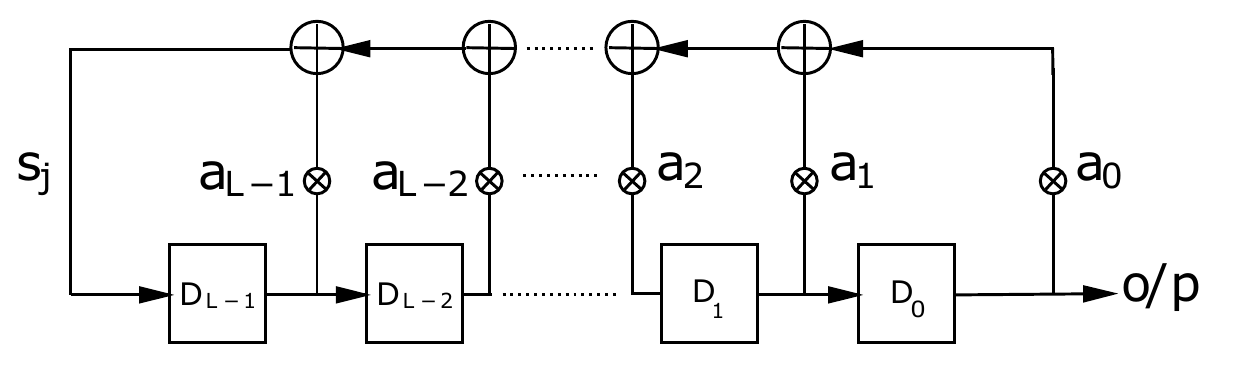}
\caption{LFSR}
\label{LFSR}
\end{figure}

The output of the LFSR shown in Figure \ref{LFSR} is a linear recurring sequence which satisfies the LRR $s_{j+L} = a_{0}{s}_{j}+a_{1}{s}_{j+1}+\ldots+a_{L-1}{s}_{j+L-1}$,
where $a_{i} \in \mathbb{F}_{q}$ for $0 \leq i \leq L-1$. With every LRR one can associate a polynomial having the same coefficients. Such a polynomial is called the characteristic polynomial of the LFSR. For example, the characteristic polynomial of the LFSR shown in Figure \ref{LFSR} is $p(x) = x^L - a_{L-1}x^{L-1} - \ldots - a_{0}$.
The degree of the characteristic polynomial is known as the degree of the LFSR. If the characteristic polynomial of an LFSR is primitive then the LFSR is called a primitive-LFSR. The outputs of the delay blocks at any given time of instant constitute the state vector of the LFSR at that instant. If the initial state is nonzero then a primitive-LFSR generates all the nonzero states in a single period \cite{lidl1997finite}. 

The linear complexity of a given sequence is the minimum degree of an LFSR which generates that sequence. Clearly, the linear complexity of a sequence generated by an LFSR is at most equal to the number of delay blocks in that LFSR. The linear complexity of such sequences can be increased by using nonlinear feedforward logic \cite{Groth1971}. An NLFG consists of an LFSR along with a multiplier assembly having a set of 2-input multipliers. 
\begin{figure}[h]
\centering
\includegraphics[scale=0.7]{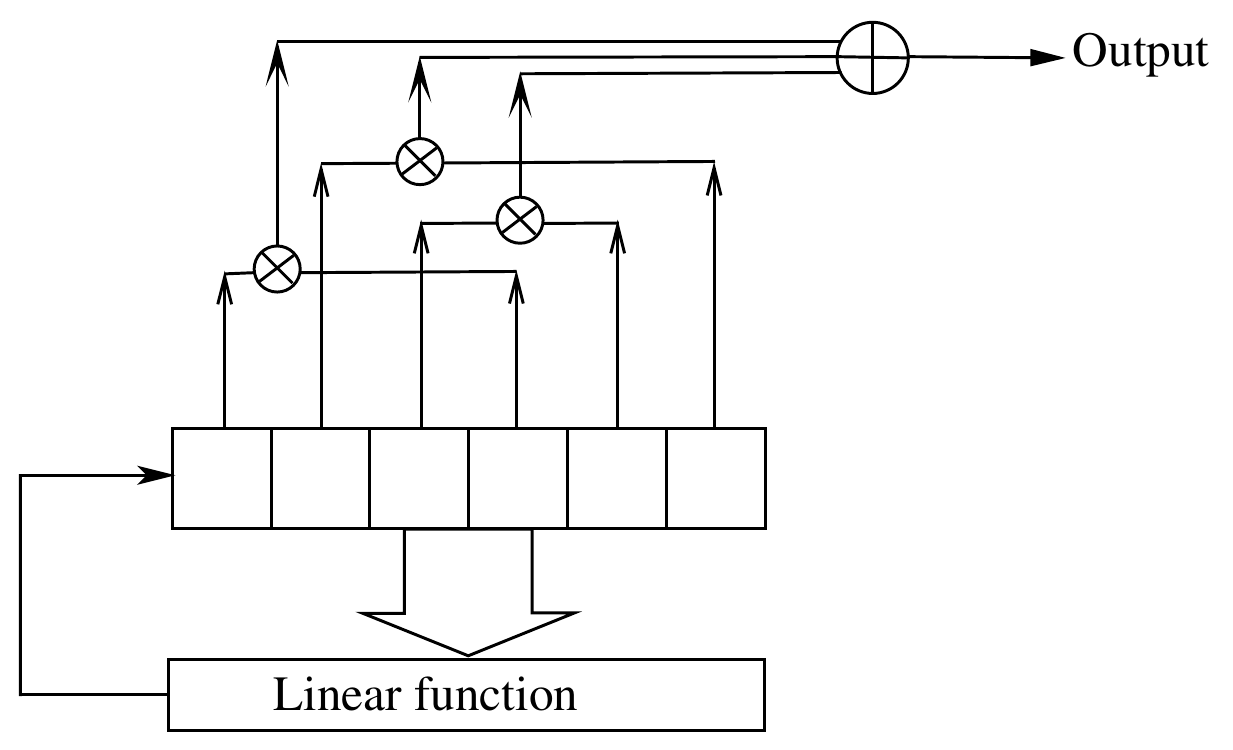} 
\caption{Nonlinear feedforward generator}
\label{NLFG}
\end{figure}

In this scheme, the output of some of the delay blocks are multiplied with each other and the resulting products are then added to generate the output sequence. The output of each delay block can act as an input to at most one multiplier. Multiplication and addition are as defined in $\mathbb{F}_{q}$. For $q = 2$, multiplication and addition translate to AND and XOR operations respectively. An example of such a scheme is shown in Figure \ref{NLFG}. In the following section, we will discuss the statistical properties of sequences generated by NLFGs over arbitrary finite fields. Our arguments do not require the underlying FSR to be linear. However, we  assume that all nonzero states occur once in every period (as in a primitive LFSR). 
\begin{figure}[h]
\centering
\includegraphics[scale=0.55]{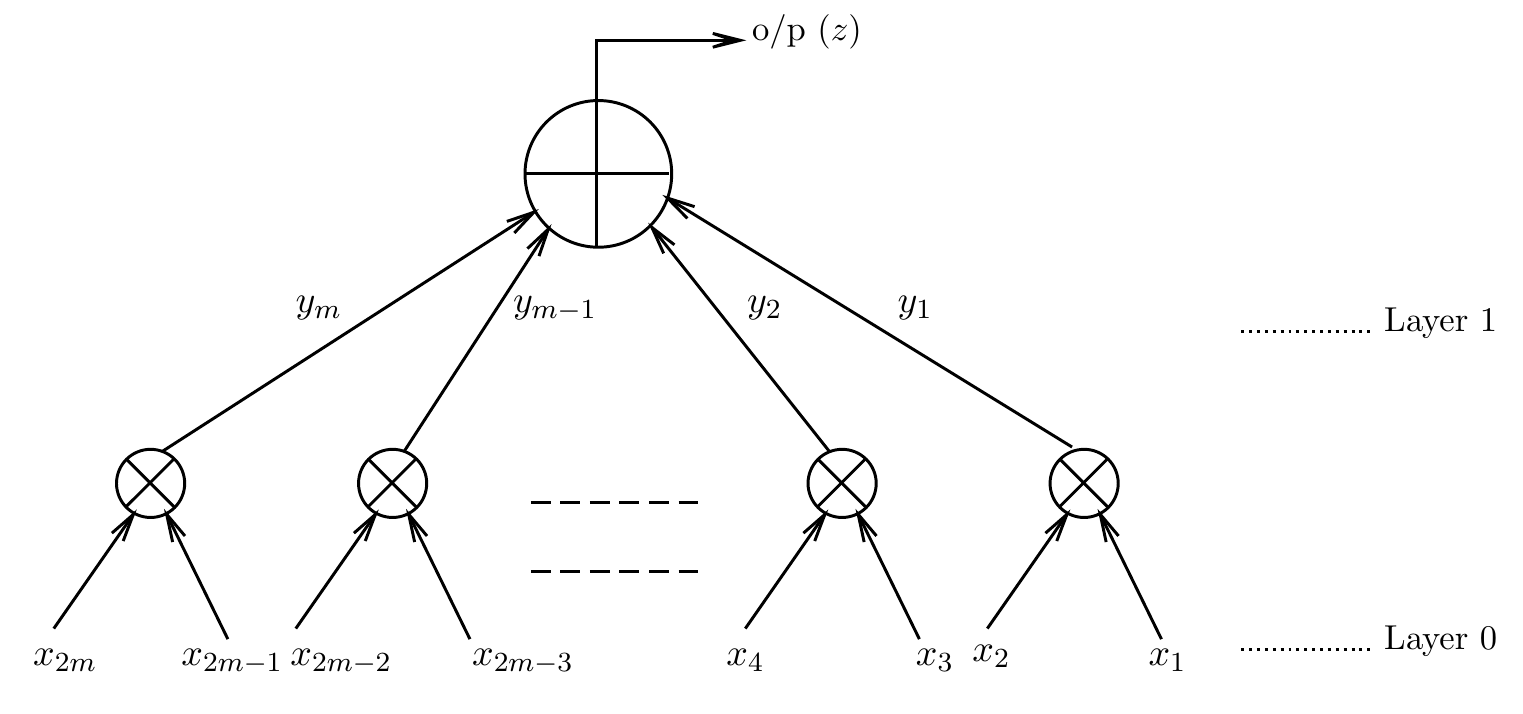}
\caption{Multiplier assembly of an NLFG with $m$ multipliers}
\label{NLFG1}
\end{figure}

\section{Statistical Properties of Sequences Generated from NLFG}
\label{sec3}
Consider an NLFG having an FSR with $L$ delay blocks and a multiplier assembly with $m \leq \lfloor \frac{L}{2} \rfloor$ multipliers. Let $\psi_{m}(K)$ denote the number of possible inputs to the multiplier assembly that generate the number $K$ at the output. When $m$ = 1, the output of the multiplier will be $0$ if either of its inputs are zero. Thus, 
\begin{equation}
\label{psi 0}
\psi_{1}(0) =2q-1~.
\end{equation}

\begin{lemma}
\label{psi 1}
$\psi_{1}(K) = (q-1)$, for all $K\in\mathbb{F}_{q}\backslash \{0\}$. 
\end{lemma}

\begin{proof}
Given any $K_{1} \in \mathbb{F}_{q}\backslash \{0\}$, there exists a unique $K_2 \in \mathbb{F}_q \backslash \{0\}$ such that $K_{1}.K_{2}=K$. Since there are $q-1$ possible values for $K_1$, $\psi_{1}(K) = (q-1)$.   
\end{proof}  

Lemma \ref{psi 1} shows that $\psi_{1}(K)$ does not depend upon the value of $K$ but only on whether $K$ is zero or nonzero. Therefore, in the remainder of the paper we denote $\psi_{1}(K)$ by $\psi_{nz}$ when $K \neq 0$ and by $\psi_{z}$ when $K=0$.   

Now, let $\mathcal{N}_{m}^{L}(K)$ be the number of nonzero state vectors of the underlying LFSR that generate $K$ at the output. Each of the $q^{2m}-1$ nonzero inputs to the multiplier assembly occurs $q^{L-2m}$ times. Therefore,

\begin{align}
\mathcal{N}_{m}^{L}(K) = 
     \begin{cases}
       \text{$q^{L-2m}.\psi_{m}(K),$} &\quad\text{when $K\neq0$} \\
       \text{$q^{L-2m}.\psi_{m}(0)-1,$} &\quad\text{when $K=0$} \\
     \end{cases}
\label{N_L_m_K}
\end{align}

In the expression for $\mathcal{N}_{m}^{L}(0)$, one is deducted to account for the absence of the zero state. Thus, deriving an expression for $\mathcal{N}_{m}^{L}(\cdot)$ reduces to finding a formula for $\psi_{m}(\cdot)$. 

\begin{definition}
An $m$ partition of $K$ over $\mathbb{F}_{q}$ is defined as an $m$-tuple of nonzero elements in $\mathbb{F}_{q}$ whose sum (as defined in $\mathbb{F}_{q}$) is $K$. We denote the set of $m$-partitions of $K$ by $\mathcal{S}_{m}(K)$.

\begin{equation*}
\mathcal{S}_{m}(K) :=
\left\{\begin{bmatrix}
y_{1}
\\ y_{2}
\\ .
\\ .
\\ y_{m}
\end{bmatrix}
\in \mathbb{F}_{q}^{m} \mid \sum_{i=1}^{m} y_{i} = K ~and~ y_{i} \neq 0 \right \}~.
\end{equation*}
where~ i = 1, 2, \ldots, m.
\end{definition}

Clearly, $|\mathcal{S}_{0}(K)| = 0 ~and~ |\mathcal{S}_{1}(K)| = 1$.
%\end{remark}
For $m \geq 1$, $|\mathcal{S}_{m}(K)|$ can be recursively calculated as follows.

\begin{lemma}
$|\mathcal{S}_{m}(K)|= \psi_{nz}^{m-1} - |\mathcal{S}_{m-1}(K)|$ where $K\in\mathbb{F}_{q}\backslash \{0\}$.
\label{Pm plus lemma}
\end{lemma}
\begin{proof}
One can arbitrarily choose $m-1$ nonzero elements from $\mathbb{F}_{q}$ in $(q-1)^{m-1}$ possible ways. If the sum of these $m-1$ elements is not equal to $K$ then there exists a unique nonzero element in $\mathbb{F}_{q}$ which gives $K$ when added with this sum. If the sum of these $m-1$ elements is equal to $K$ then this ($m-1$)-tuple is a member of the set $\mathcal{S}_{m-1}(K)$. Hence, $|\mathcal{S}_{m}(K)|= (q-1)^{m-1} - |\mathcal{S}_{m-1}(K)| = \psi_{nz}^{m-1} - |\mathcal{S}_{m-1}(K)|$. 
\end{proof} 

Using the above recursion, the closed-form expression for $|\mathcal{S}_{m}(K)|$ is derived as follows.
\begin{lemma}
\label{Pm}
$|\mathcal{S}_{m}(K)| = \frac{1}{q}.\{\psi_{nz}^m - (-1)^{m}\}$, where $K\in\mathbb{F}_{q}\backslash \{0\}$.
\end{lemma}
\begin{proof} 
We shall prove the lemma using induction.
\\Now, $|\mathcal{S}_{1}(K)| = 1 = \frac{1}{q}(\psi_{nz} + 1)$. Thus, the statement of the lemma is true for $m=1$.

Let the statement be true for $m = l$, i.e., $|\mathcal{S}_{l}(K)| = \frac{1}{q}.\{\psi_{nz}^l - (-1)^{l}\}$.
We now proceed to prove that the statement is true for $m = l+1$.
\begin{align*}
|\mathcal{S}_{l+1}(K)| &= \psi_{nz}^{l} - |\mathcal{S}_{l}(K)| ~~[using ~lemma ~\ref{Pm plus lemma}]
\\&= \psi_{nz}^{l} - \frac{1}{q}\{\psi_{nz}^l - (-1)^{l}\}
%\\&= \psi_{nz}^{l} - \frac{\psi_{nz}^l}{q} + \frac{(-1)^l}{q}
\\&= \frac{1}{q}\{q\psi_{nz}^{l} - \psi_{nz}^l + (-1)^l\}
\\&= \frac{1}{q}\{\psi_{nz}^{l}(q-1) + (-1)^l\}
\\&= \frac{1}{q}\{\psi_{nz}^{l+1} - (-1)^{l+1}\}~~[since~\psi_{nz}=q-1]
\end{align*}
\end{proof}

Assume that at a particular time instant, the outputs of $i$ of the $m$ multipliers are zero. These $i$ multipliers can be chosen in $\binom{m}{i}$ ways. Each of these multipliers can have $\psi_z$ possible pairs of inputs. Now, there are $|\mathcal{S}_{m-i}(K)|$ possible sets of outputs from the remaining $m-i$  multipliers such that the output of the adder is $K$. For each such set each multiplier can have $\psi_{nz}$ possible pairs of inputs. Therefore,

\begin{align}
\label{Result}
\psi_{m}(K)=\sum\limits_{i=0}^{m-1}\binom{m}{i}\psi_{z}^{i}\psi_{nz}^{m-i}|\mathcal{S}_{m-i}(K)|,~where~ K\in\mathbb{F}_{q}\backslash \{0\}.
\end{align}

Now, we simplify the above above formula to derive a closed form expression for $\psi_{m}(K)$.

\begin{theorem}
\label{main result}
For a multiplier assembly with $m$ multipliers and for all $K\in\mathbb{F}_{q}$. 
\begin{align*}
\psi_{m}(K) = 
     \begin{cases}
       \text{$q^{{m-1}}(q^{m}-1),$} &\quad\text{when $K\neq0$} \\
       \text{$q^{m-1}(q^{m}+q-1),$} &\quad\text{when $K=0$} \\
     \end{cases}
\end{align*}
\end{theorem}

\begin{proof}
Let $K \neq 0$. Substituting the formula for $|\mathcal{S}_{m-i}(K)|$ from Lemma \ref{Pm} in Equation \ref{Result} we get - 
\begin{align*}
\psi_{m}(K) &= \frac{1}{q}\sum\limits_{i=0}^{m-1}\binom{m}{i}\psi_{z}^{i}\psi_{nz}^{m-i}\{\psi_{nz}^{m-i}-(-1)^{m-i}\}
\\ &= \frac{1}{q}\bigg\{\sum\limits_{i=0}^{m-1}\binom{m}{i}\psi_{z}^{i}\psi_{nz}^{2(m-i)}-\sum\limits_{i=0}^{m-1}\binom{m}{i}\psi_{z}^{i}(-\psi_{nz})^{(m-i)}\bigg\}
\end{align*}
Now, $\sum\limits_{i=0}^{m-1}\binom{m}{i}\psi_{z}^{i}\psi_{nz}^{2(m-i)} = \sum\limits_{i=0}^{m}\binom{m}{i}\psi_{z}^{i}\psi_{nz}^{2(m-i)}-\psi_{z}^{m}$ and $\sum\limits_{i=0}^{m-1}\binom{m}{i}\psi_{z}^{i}(-\psi_{nz})^{(m-i)} = \sum\limits_{i=0}^{m}\binom{m}{i}\psi_{z}^{i}(-\psi_{nz})^{(m-i)}-\psi_{z}^{m}$. 
\\~\\
Therefore,
\begin{align*}
\psi_{m}(K) &= \frac{1}{q}\bigg[\bigg\{\sum\limits_{i=0}^{m}\binom{m}{i}\psi_{z}^{i}\psi_{nz}^{2(m-i)}-\psi_{z}^{m}\bigg\}
\\&-\bigg\{\sum\limits_{i=0}^{m}\binom{m}{i}\psi_{z}^{i}(-\psi_{nz})^{(m-i)}-\psi_{z}^{m}\bigg\}\bigg]
\\&= \frac{1}{q}\{(\psi_{z}+\psi_{nz}^{2})^{m}-(\psi_{z}-\psi_{nz})^{m}\}
\intertext{Substituting the values of $\psi_{z}$ and $\psi_{nz}$ from Equation \ref{psi 0} and Lemma \ref{psi 1} we get - } 
\psi_{m}(K) &=\frac{1}{q}\bigg[\{(2q-1)+(q-1)^{2}\}^{m}-\{(2q-1)-(q-1)\}^{m}\bigg]
\\&=\frac{1}{q}(q^{2m}-q^m)=\frac{q^{m}}{q}(q^{m}-1)
\\&=q^{m-1}(q^{m}-1)
\end{align*}

Since there are $(q-1)$ nonzero elements in $\mathbb{F}_q$, there are $(q-1)q^{m-1}(q^{m}-1)$ input combinations that generate a nonzero output from the NLFG. Therefore, 

\begin{align*}
\psi_{m}(0)& = q^{2m}-(q-1).q^{m-1}(q^{m}-1)
\\ & = q^{2m-1}+q^{m}-q^{m-1}
\\&= q^{m-1}(q^{m}+q-1)
\end{align*}

This concludes the proof of our theorem.  

\end{proof}

Substituting the formula for $\psi_{m}(\cdot)$ derived in Theorem \ref{main result} in Equation \ref{N_L_m_K} we get -   

\begin{corollary}
$\newline$
$\mathcal{N}_{m}^{L}(K)=\left\{
\begin{array}{ll}
       q^{L-m-1}(q^{m}-1) & ~~where ~K \neq 0.\\
       q^{L-m-1}(q^{m}+q-1)-1 & ~~where~ K=0.\\
\end{array}
\right. $
\label{N_L_m_K_final}
\end{corollary}

\begin{remark}
It can be easily verified that $\mathcal{N}_{m}^{L}(0)+ (q-1)\mathcal{N}_{m}^{L}(K)= q^{L}-1$.
\end{remark}

\begin{remark}
The Theorem 3 in \cite{Dawson1990} is a special case of the Theorem \ref{main result} where $q=2$.
\end{remark}

We now go on to show that the distribution of elements in the output sequence of an NLFG tends to a balanced distribution as the number of delay blocks and the number of multipliers tends to infinity.   

\begin{corollary}
\label{bd1}
$$\lim_{m \to \infty} \frac{\mathcal{N}_{m}^{L}(K)}{q^{L}-1} = \frac{1}{q}, ~~~where~ K\in \mathbb{F}_{q} .$$ 
\end{corollary}

\begin{proof}
In the case, when $K \neq 0$ then -  
\begin{align*}
\lim_{m \to \infty} \frac{\mathcal{N}_{m}^{L}(K)}{q^{L}-1} &=\lim_{m \to \infty} \frac{q^{L-m-1}(q^{m}-1)}{q^{L}-1} = \frac{1}{q}.
\end{align*}

In the case, when $K = 0$ then -  
\begin{align*}
\lim_{m \to \infty} \frac{\mathcal{N}_{m}^{L}(0)}{q^{L}-1} &=\lim_{m \to \infty} \frac{q^{L-m-1}(q^{m}+q-1)-1}{q^{L}-1} \\&= \lim_{m \to \infty}\frac{q^{L}}{(q^{L}-1)}.\frac{q^{L-1}+q^{L-m}-q^{L-m-1}-1}{q^{L}}\\
&=\lim_{m \to \infty} \frac{q^{L}}{(q^{L}-1)}.\lim_{m \to \infty}\bigg[ \frac{1}{q}+\frac{1}{q^{m}}-\frac{1}{q^{m+1}}-\frac{1}{q^L}\bigg ] \\
&= \frac{1}{q}
\end{align*}
\end{proof}

\section{NLFGs over $\sigma$-LFSR}
\label{sec 4}

A $\sigma$-LFSR is an LFSR configuration with multi-input multi-output delay blocks that aims to utilize the parallelism provided by modern word based processors. A detailed description of $\sigma$-LFSRs can be found in \cite{zeng2007high}. Figure \ref{sigma_LFSR} depicts an $L$-stage $\sigma$-LFSR with $r$-input $r$-output delay blocks.
\begin{figure}[h]
\centering
\includegraphics[scale=0.8]{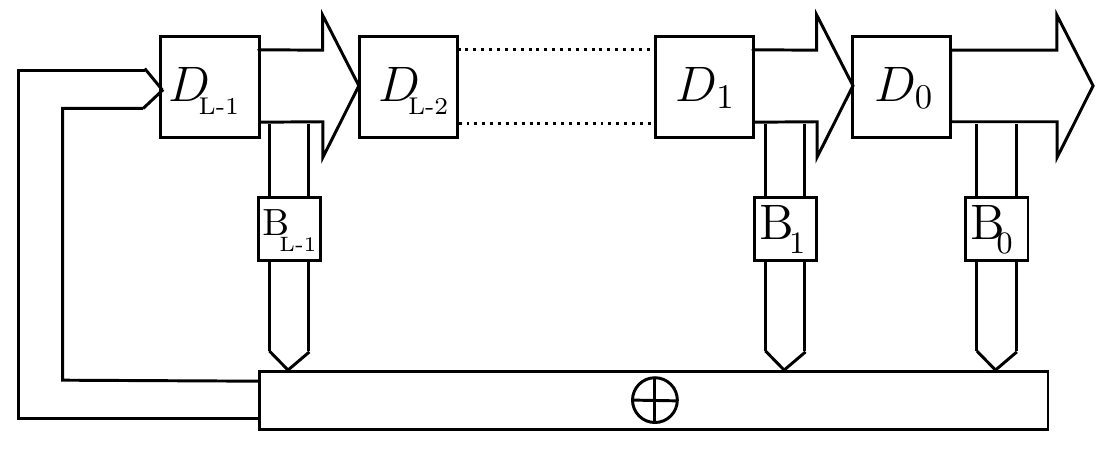}
\caption{r-input, r-output $\sigma$-LFSR of order $L$ over $\mathbb{F}_{q^{r}}$}
\label{sigma_LFSR}
\end{figure}

The feedback gain matrices $B_{0}$, $B_{1}$, \ldots, $B_{L-1}$ are elements in $\mathbb{F}_{q}^{r \times r}$. The output sequence of a $\sigma$-LFSR satisfies the following linear recurring relation
\begin{align}
\label{sigma LFSR}
\textbf{s}_{j+L} = B_{0}\textbf{s}_{j}+B_{1}\textbf{s}_{j+1}+\ldots+B_{L-1}\textbf{s}_{j+L-1}
\end{align}
where $j$=0,1,\ldots and $\textbf{s}_{j} \in \mathbb{F}_{q}^{r}$. At the $k$-th time instant, let $\textbf{s}_{i}(k)$ be the output of the $B_{i}$-th delay block. The state vector $\textbf{s}(k)$ of an $\sigma$-LFSR at that instant can be obtained by stacking the outputs of the delay blocks one below the other. For instance, 
\begin{align*}
\textbf{s}(k) =
\begin{bmatrix}
\textbf{s}_{0}(k)\\ 
\textbf{s}_{1}(k)\\ 
\vdots\\ 
\textbf{s}_{L-1}(k)\\ 
\end{bmatrix}
\in \mathbb{F}_{q}^{rL}
\end{align*}
Observe that, 
\begin{align*}
\textbf{s}_{0}(k+1)&=\textbf{s}_{1}(k)\\
\textbf{s}_{1}(k+1)&=\textbf{s}_{2}(k)\\
&\vdots\\
\textbf{s}_{L-2}(k+1)&=\textbf{s}_{L-1}(k)\\
\textbf{s}_{L-1}(k+1)&=B_{0}\textbf{s}_{0}(k)+B_{1}\textbf{s}_{1}(k)+\ldots+B_{L-1}\textbf{s}_{L-1}(k).
\end{align*} 
Thus, the relation between two consecutive state vectors of a $\sigma$-LFSR is as follows:
\begin{align}
\textbf{s}(k+1)=A_{rL}\textbf{s}(k), ~~~~\forall~k=0,1,...
\end{align} 
where

\begin{align*}
A_{rL}=
\left [ \begin{matrix}
0 & I & 0 & \hdots  & 0\\ 
0 & 0 & I& \hdots  & 0\\ 
\vdots & \vdots &  \vdots & \ddots &\vdots\\ 
0 & 0 & 0& \hdots  & I\\ 
B_{0} & B_{1} & B_{2} & \hdots & B_{L-1}
\end{matrix} \right ] \in \mathbb{F}_{q}^{rL \times rL}
\end{align*} 
Here, $0\in \mathbb{F}_{q}^{r \times r}$ is the zero matrix and $I\in \mathbb{F}_{q}^{r \times r}$ is the identity matrix. The matrix $A_{rL}$ is called the state transition matrix of the $\sigma$-LFSR. The characteristic polynomial of the state transition matrix is called the characteristic polynomial of the $\sigma$-LFSR. As in a conventional LFSR, if the characteristic polynomial of the $\sigma$-LFSR is primitive then all nonzero states are covered in a single period. 
%Such $\sigma$-LFSRs are called primitive $\sigma$-LFSRs.% 
Given positive integers $r$ and $L$ and a primitive polynomial $p(x)$ of degree $rL$, the number of $\sigma$-LFSR configurations having characteristic polynomial $p(x)$ has been calculated in \cite{Krishnaswamy2012}, \cite{srinitheis}. 

The output sequence of a $\sigma$-LFSR with $r$-input $r$-output delay blocks is a sequence in $\mathbb{F}_q^r$.  Now, each entry of this vector sequence constitutes a scalar sequence. We shall call these sequences the component sequences of the vector sequence. 
\begin{lemma}
\label{component sequence}
Each component sequence of a vector sequence generated by a primitive $\sigma$-LFSR has the same characteristic polynomial as that of the $\sigma$-LFSR.
\end{lemma}
\begin{proof}
Consider a $\sigma$-LFSR with $L$ $r$-input $r$-output delay blocks. Let $p(x) = x^{rL}-p_{rL-1}x^{rL-1}-p_{rL-2}x^{rL-2}-\ldots-p_0$ be its primitive characteristic polynomial and $A$ be its state transition matrix. If the initial state vector is $v_0$ then the sequence of state vectors is given by $(v_{i})_{i=0}^{\infty} = \{v_0, Av_0, A^2v_0, \ldots \}$. Given any state vector $v \in \mathbb{F}_q^{rL}$, $p(A)v = 0$. Therefore, the sequence of state vectors satisfies the following LRR. 
\begin{align}
\label{LRR}
v_{j+L}=p_0v_{j}+p_1v_{j+1}+\ldots+p_{L-1}v_{j+L-1}
\end{align}
where $j=0, 1, \ldots$. Clearly, each entry of the state vector obeys the above LRR. Therefore, each component sequence satisfy the LRR. Consequently, the characteristic polynomial of each component sequence divides $p(x)$. Since $p(x)$ is primitive, this is possible only if each of these polynomials is $p(x)$.
\end{proof}

Since $\mathbb{F}_{q}^{r}$ is known to be isomorphic to $\mathbb{F}_{q^r}$,  a $\sigma$-LFSR can be seen as an FSR over the field $\mathbb{F}_{q^r}$ \cite{lidl1997finite}. Thus, each state vector of a $\sigma$-LFSR can be seen as a vector in $\mathbb{F}_{q^r}^L$. The characteristic polynomial of the $\sigma$-LFSR being primitive ensures that all non zero vectors in $\mathbb{F}_{q^r}^{L}$ occur as state vectors exactly once in every period. In the proposed scheme, the outputs of delay blocks of a $\sigma$-LFSR are multiplied as elements in $\mathbb{F}_{q^r}$. This is in contrast to the scheme given in \cite{sartaj2012} wherein multiplication is done element-wise. Note that element-wise multiplication is not equivalent to multiplication over a finite field. For example, in $\mathbb{F}_{2}^{4}$  the element-wise product of two nonzero vectors $v_{1}=[1 0 0 1]^{T},~ v_{2}=[0 1 1 0]^{T}$ is zero which is not possible over a finite field.

\begin{figure*}[htb]
\centering
\includegraphics[scale=0.8]{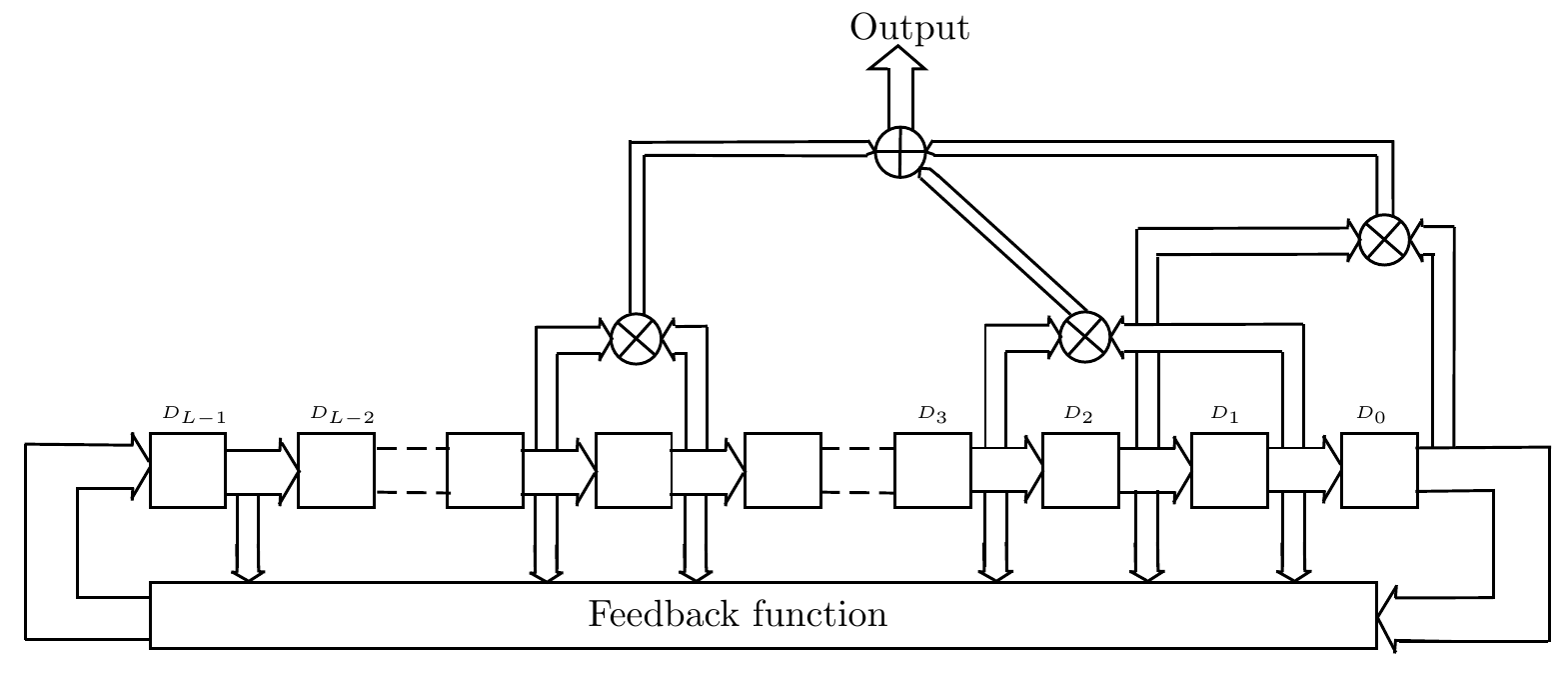}
\caption{NLFG based on $\sigma$-LFSR}
\label{sigma NLFG}
\end{figure*}

Let $p(x)$ be a primitive polynomial of degree $r$. Now, $\mathbb{F}_{q^r}$ can be seen as the residue class ring $^{\mathbb{F}_{q}[x]}/_{<p(x)>}$. The set $\{[1],[x],\ldots,[x^{r-1}]\}$ is a basis of $^{\mathbb{F}_{q}[x]}/_{<p(x)>}$, where $[x]$ denotes the equivalence class of $x$. Given a polynomial $f(x) \in \mathbb{F}_{q}[x]$, the equivalence class of $f(x)$ has a unique representative element with degree less than $r$. We therefore have the following map $\mathcal{M}: \mathbb{F}_{q^r} \rightarrow \mathbb{F}_{q}^{r}$. 

\begin{align*}
\mathcal{M}\big(f_0[1]+f_1[x]+\ldots+f_{r-1}[x^{r-1}]\big)=\begin{bmatrix}
f_0\\
f_1\\
\vdots\\
f_{r-1}
\end{bmatrix}
\end{align*}
Clearly, the above map is a vector space homomorphism. Using this map, we define multiplication of two elements in $\mathbb{F}_{q}^{r}$, denoted as $\times$, as follows. 
\begin{align*}
v_1 \times v_2 = \mathcal{M}\big([\mathcal{M}^{-1}(v_1).\mathcal{M}^{-1}(v_2)\big])
\end{align*}
where $v_1,v_2 \in \mathbb{F}_{q}^{r}$. Let $v_1=\mathcal{M}([f_{1}(x)])$ and $v_2=\mathcal{M}([f_{2}(x)])$. Therefore, $v_1 \times v_2$ is a vector whose entries are the coefficients of the polynomial $g(x)=f_1f_2~ mod ~p(x)$. If $f_1$ and $f_2$ are the unique elements in their respective equivalence classes having degree less than $r$ then $f_1(x)f_2(x)$ is a polynomial with degree less than $2r$. Let $v \in \mathbb{F}_q^{2r-1}$ be a vector whose entries are the coefficients of $f_1f_2$. Now, $v_{f_1f_2} = v_1 \ast v_2$ where $\ast$ denotes convolution. Observe that $v_1 \times v_2 = \mathcal{Q}v_{f_1f_2} = \mathcal{Q}(v_1 \ast v_2)$ where $\mathcal{Q}\in \mathbb{F}_q^{(2r-1)\times(2r-1)}$ is the following matrix.

\begin{align}
\mathcal{Q} = [I^{(r-1)\times (r-1)}~\vdots~\mathcal{M}([x^{r}])~\ldots~\mathcal{M}([x^{2r-2}])~\vdots~\mathcal{M}([x^{2r-1}])]
\label{Qmatrix}
\end{align}
%Thus, $v_1 \boxtimes v_2 = \mathcal{Q}(v_1 \ast v_2)$.

\begin{example}
Consider vectors $v_1 = [1~1~0]^{T}$, $v_2 = [1~0~1]^{T} \in \mathbb{F}_{2}^{3}$. Let $p(x)=x^3+x+1$. From Equation \ref{Qmatrix}, the $Q$ matrix is as follows.
\begin{align*}
 \mathcal{Q}=
\begin{bmatrix}
1~ 0~ 0~ 1~ 0 \\0~ 1~ 0~ 1~ 1 \\0~ 0~ 1~ 0~ 1
\end{bmatrix}_{3 \times 5}
\end{align*}
Now, $v_1 \ast v_2 = [1~1~1~1~0]^{T} \in \mathbb{F}_{2}^{5}$. Therefore, $v_1 \times v_2 = \mathcal{Q}(v_1 \ast v_2) = [0~0~1]^{T}$.
\end{example}

As shown in Figure \ref{sigma NLFG}, in the proposed scheme the underlying FSR is a $\sigma$-LFSR and the multiplier assembly has $m \leq \lfloor \frac{L}{2}\rfloor$ multipliers. 
Each multiplier takes the output of two distinct $r$-input $r$-output delay blocks, convolves them and multiplies the result with the matrix $Q$ given in Equation \ref{Qmatrix}. It thus implements the map `$\times$' described above. The outputs of the multipliers are then added to generate the output vector sequence. As in a conventional NLFG, the output of each delay block can act as an input to at most one multiplier. Since the proposed scheme views a $\sigma$-LFSR as an FSR over $\mathbb{F}_{q^r}$ and the outputs of the delay blocks are multiplied as elements of $\mathbb{F}_{q^r}$, the analysis given in Section \ref{sec3} is valid for this scheme. Let $\mathbf{N}_{m}^{L}(v)$ be the number of occurrences of a vector $v$ in a single cycle of the sequence generated by the proposed NLFG. From Corollary \ref{N_L_m_K_final}, $\mathbf{N}_{m}^{L}(v)$ is given by; 
\begin{align}
\mathbf{N}_{m}^{L}(v) =
\left\{
\begin{array}{ll}
       q^{r(L-m-1)}(q^{rm}-1) & ~~when~ v \neq 0\\
       q^{r(L-m-1)}(q^{rm}+q^{r}-1)-1 & ~~when~ v=0\\
\end{array}
\right.
\label{eqn10}
\end{align}

In order to draw a comparison between the proposed scheme and that given in  \cite{sartaj2012}, we now briefly analyse the distribution of vectors in sequences generated by the latter. Although \cite{sartaj2012} deals only with the binary case, in our analysis we consider the NLFG to be over an arbitrary finite field $\mathbb{F}_q$. The only difference between the scheme given in \cite{sartaj2012} and the one proposed here is that there the output of the delay blocks are multiplied element-wise. In the remainder of this section, we shall refer to NLFGs that use the scheme given in \cite{sartaj2012}  as element-wise NLFGs. Element-wise multiplication operation in a multiplier assembly is depicted in Figure \ref{element wise adder}.

\begin{figure}[h]
\centering
\includegraphics[scale=0.75]{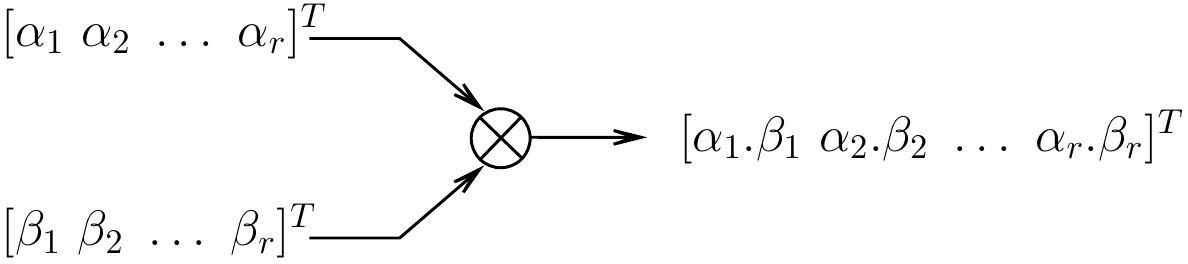}
\caption{Element-wise addition and multiplication}
\label{element wise adder}
\end{figure}

\begin{theorem}
\label{main result for mimo ma}
Consider an element-wise NLFG having $L$ $r$-input $r$-output delay blocks and $m \leq \lfloor \frac{L}{2} \rfloor$ multipliers. For a given nonzero vector $\mathbf{v}\in \mathbb{F}_q^r$, the number $\mathbf{\Psi_{m}(v)}$ of inputs to the multiplier assembly that generate $\mathbf{v}$ at the output is given by  $$\mathbf{\Psi_{m}(v)}= (q^{m-1})^{r}(q^{m}-1)^{\kappa}(q^{m}+q-1)^{r-\kappa}$$ where $\kappa$ is the number of nonzero elements in $v$.  
\end{theorem}

\begin{proof}
Since addition and multiplication are performed element-wise, the $i$-th entry $v_i$ of the output vector sequence is a function of only the $i$-th outputs of the delay blocks of the $\sigma$-LFSR. Further, from Lemma \ref{component sequence} it can be inferred that each component sequence of the $\sigma$-LFSR can be seen to be generated by a scalar LFSR whose characteristic polynomial is the same as that of the $\sigma$-LFSR. Therefore, the $i$-th bit of the output sequence of the NLFG can be seen to be generated by a scalar NLFG with a primitive scalar LFSR having $rL$ delay blocks and a multiplier assembly with $m$ multipliers. From Theorem \ref{main result}, the number of inputs to this multiplier assembly that generates $v_i$ at the output is given by 

\begin{align*}
\psi_{m}(v_i)=
\left\{
\begin{array}{ll}
       q^{m-1}(q^{m}-1) & ~~~~when~ v_i \neq 0.\\
       q^{m-1}(q^{m}+q-1) & ~~~~when~ v_i=0.\\
\end{array}
\right.
\end{align*}
Therefore, the total number of possible inputs to the multiplier assembly that generates a given vector $v$ having $\kappa$ nonzero elements is given by 
\begin{align*}
\mathbf{\Psi_{m}(v)} 
&=\big\{q^{{m-1}}(q^{m}-1)\big\}^{\kappa}\big\{q^{{m-1}}(q^{m}+q-1)\big\}^{r-\kappa}\\
&= (q^{m-1})^{r}(q^{m}-1)^{\kappa}(q^{m}+q-1)^{r-\kappa}
\end{align*}
\end{proof}

\begin{remark}
Clearly, in the case when $r=1, \kappa=1$ and $r=1, \kappa=0$, Theorem \ref{main result for mimo ma} translates to Theorem \ref{main result} .   
\end{remark}

For an NLFG having $L$ $r$-input $r$-output delay blocks and $m \leq \lfloor L/2 \rfloor$ multipliers, let $\mathfrak{N}_{m}^{L}(v)$ denote the number of times in a single cycle that the vector $v\in \mathbb{F}_{q}^{r}$ occurs at the output of the NLFG. 

\begin{corollary}
$\newline$
\label{N_M_dot}
$\mathfrak{N}_{m}^{L}(v) =
\left\{
\begin{array}{ll}
      q^{r(L-m-1)}(q^{m}-1)^{\kappa}(q^m+q-1)^{r-\kappa} & ~ ,v \neq 0.\\
      q^{r(L-m-1)}(q^m+q-1)^{r}-1 &~ ,v=0.\\
\end{array}
\right. 
$
\end{corollary}
\begin{proof}
Since every nonzero state vector occurs exactly once in every period of the underlying primitive $\sigma$-LFSR, $\mathfrak{N}_{m}^{L}(v)$ is equal to the number of nonzero states of the $\sigma$-LFSR that generate $v \in \mathbb{F}_q^r$ at the output of the NLFG. Clearly, for each input to the multiplier assembly there are $q^{L-2m}$ possible state vectors of the $\sigma$-LFSR (since $L-2m$ of the delay blocks are not connected to the multiplier assembly). Therefore, the number of times a nonzero vector $v$ occurs at the output of the NLFG in a single period is equal to $q^{r(L-2m)}\mathbf{\Psi_{m}(v)}$. Now, among the states of the $\sigma$-LFSR that result in zero at the output of the NLFG is the zero state. However, this state does not occur in any nonzero cycle. Therefore, the number of times the zero vector occurs at the output of the NLFG in a single period is equal to $q^{r(L-2m)}\mathbf{\Psi_{m}(0)}-1$. Thus, 
$$
\mathfrak{N}_{m}^{L}(v) =
\left\{
\begin{array}{ll}
      q^{r(L-2m)}\mathbf{\Psi_{m}(v)} & ~~when~ v \neq 0.\\
      q^{r(L-2m)}\mathbf{\Psi_{m}(0)}-1 & ~~when~ v=0.\\
\end{array}
\right. 
$$
Substituting the value of $\Psi_{m}(v)$ from Theorem \ref{main result for mimo ma}, we get-
$$
\mathfrak{N}_{m}^{L}(v) =
\left\{
\begin{array}{ll}
       q^{r(L-m-1)}(q^{m}-1)^{\kappa}(q^m+q-1)^{r-\kappa} & ~~ ,v \neq 0.\\
      q^{r(L-m-1)}(q^m+q-1)^{r}-1 & ~~, v=0.\\
\end{array}
\right.
$$ 
\end{proof}

Comparing the formulae derived in Corollary \ref{N_M_dot} with those in Equation \ref{eqn10}, it is clearly seen that the output sequence of an element-wise NLFG has a bias towards vectors having a greater number of zeros. This however is not the case with the scheme proposed in this paper.

\begin{example}
Let $q=2, L=5, m=2$ and $r=3$. The number of occurrences of $v_1= [0~0~0]^T$, $v_2=[0~1~0]^T$ and $v_3=[1~1~1]^T$ at the output of an element-wise NLFG are $7999, 4800$ and $1728$ respectively. However, the number of occurrences of the vectors $v_1, v_2$ and $v_3$ at the output of our proposed NLFG scheme are $4543, 4032$ and $4032$ respectively. 
\end{example}

\section{Conclusion}
\label{sec5}
In this paper, we have extended the notion of NLFGs to arbitrary finite fields and have analyzed the statistical properties of the sequences generated by such NLFGs. Further, we have proposed an implementation of NLFGs over $\sigma$-LFSRs and have shown that the sequences generated by such proposed scheme are more balanced than the sequences generated by the existing scheme given in \cite{sartaj2012}.

\section*{Acknowledgment}

The authors are grateful to Prof. Harish K. Pillai, Department of Electrical Engineering, Indian Institute of Technology Bombay, without whom this work would never have been possible.

% or


%\appendices
%% you can choose not to have a title for an appendix
%% if you want by leaving the argument blank
%\section{}



\ifCLASSOPTIONcaptionsoff
  \newpage
\fi



%\bibliographystyle{IEEEtran}

\bibliography{typeinst}

\begin{thebibliography}{10}
\providecommand{\url}[1]{#1}
\csname url@samestyle\endcsname
\providecommand{\newblock}{\relax}
\providecommand{\bibinfo}[2]{#2}
\providecommand{\BIBentrySTDinterwordspacing}{\spaceskip=0pt\relax}
\providecommand{\BIBentryALTinterwordstretchfactor}{4}
\providecommand{\BIBentryALTinterwordspacing}{\spaceskip=\fontdimen2\font plus
\BIBentryALTinterwordstretchfactor\fontdimen3\font minus
  \fontdimen4\font\relax}
\providecommand{\BIBforeignlanguage}[2]{{%
\expandafter\ifx\csname l@#1\endcsname\relax
\typeout{** WARNING: IEEEtran.bst: No hyphenation pattern has been}%
\typeout{** loaded for the language `#1'. Using the pattern for}%
\typeout{** the default language instead.}%
\else
\language=\csname l@#1\endcsname
\fi
#2}}
\providecommand{\BIBdecl}{\relax}
\BIBdecl

\bibitem{Paar2009}
C.~Paar and J.~Pelzl, \emph{Understanding Cryptography: A Textbook for Students
  and Practitioners}.\hskip 1em plus 0.5em minus 0.4em\relax Springer Berlin
  Heidelberg, 2009.

\bibitem{menezes1996}
A.~Menezes, P.~van Oorschot, and S.~Vanstone, \emph{Handbook of Applied
  Cryptography}, ser. Discrete Mathematics and Its Applications.\hskip 1em plus
  0.5em minus 0.4em\relax CRC Press, 1996.

\bibitem{peterson1972error}
W.~Peterson and E.~Weldon, \emph{Error-correcting Codes}.\hskip 1em plus 0.5em
  minus 0.4em\relax MIT Press, 1972.

\bibitem{1095533}
R.~Pickholtz, D.~Schilling, and L.~Milstein, ``Theory of spread-spectrum
  communications--a tutorial,'' \emph{IEEE Trans. on Comm.}, vol.~30, no.~5,
  pp. 855--884, May 1982.

\bibitem{Golomb:1981:SRS:578271}
S.~W. Golomb, \emph{Shift Register Sequences}.\hskip 1em plus 0.5em minus
  0.4em\relax Laguna Hills, CA, USA: Aegean Park Press, 1981.

\bibitem{Groth1971}
E.~Groth, ``Generation of binary sequences with controllable complexity,''
  \emph{IEEE Trans. on Inf. Theory}, vol.~17, no.~3, pp. 288--296, May 1971.

\bibitem{key1976}
E.~KEY, ``An analysis of the structrue and complexity of nonlinear binary
  sequence generators,'' \emph{IEEE Trans. on Inf. Theory}, vol.~22, no.~6, pp.
  732--736, 1976.

\bibitem{Dawson1990}
E.~Dawson, J.~Asenstorfer, and P.~Gray, ``Cryptographic properties of groth
  sequences,'' \emph{Australasian Journal of Combinatorics}, vol.~1, pp.
  53--65, 1990.

\bibitem{bedi2001}
S.~Bedi and N.~Pillai, ``\BIBforeignlanguage{English}{Cryptanalysis of the
  nonlinear feedforward generator},'' in
  \emph{\BIBforeignlanguage{English}{Progress in Cryptology — INDOCRYPT 2001}},
  ser. Lecture Notes in Computer Science, C.~Rangan and C.~Ding, Eds.\hskip 1em
  plus 0.5em minus 0.4em\relax Springer Berlin Heidelberg, 2001, vol. 2247, pp.
  188--194.

\bibitem{gammel2006}
B.~M. Gammel and R.~G{\"o}ttfert, ``Linear filtering of nonlinear
  shift-register sequences,'' in \emph{Coding and Cryptography}.\hskip 1em plus
  0.5em minus 0.4em\relax Springer, 2006, pp. 354--370.

\bibitem{teo2013}
S.~G. Teo, ``Analysis of nonlinear sequences and streamciphers,'' Ph.D.
  dissertation, Queensland University of Technology, 2013.

\bibitem{sartaj2012}
S.~U. Hasan, D.~Panario, and Q.~Wang, ``Word-oriented transformation shift
  registers and their linear complexity,'' in \emph{Sequences and Their
  Applications -- SETA 2012}, ser. Lecture Notes in in Computer Science,
  T.~Helleseth and J.~Jedwab, Eds., vol. 7280.\hskip 1em plus 0.5em minus
  0.4em\relax Berlin, Heidelberg: Springer Berlin Heidelberg, 2012, pp.
  190--201.

\bibitem{lidl1997finite}
R.~Lidl and H.~Niederreiter, \emph{Finite Fields}, ser. Encyclopedia of
  Mathematics and its Applications.\hskip 1em plus 0.5em minus 0.4em\relax
  Cambridge University Press, 1997, no. v. 20, pt. 1.

\bibitem{zeng2007high}
G.~Zeng, W.~Han, and K.~He, ``High efficiency feedback shift register:
  $\sigma$-lfsr.'' \emph{IACR Eprint archive}, 2007.

\bibitem{Krishnaswamy2012}
S.~Krishnaswamy and H.~K. Pillai, ``On the number of linear feedback shift
  registers with a special structure,'' \emph{{IEEE} Transactions on
  Information Theory}, vol.~58, no.~3, pp. 1783--1790, 2012.

\bibitem{srinitheis}
S.~Krishnaswamy, ``On multisequences and applications,'' Ph.D. dissertation,
  Indian Institute of Technology Bombay, 2012.

\end{thebibliography}
\bibliographystyle{ieeetran}

%\begin{thebibliography}{1}
%\end{thebibliography}


%\begin{IEEEbiography}{Suman Roy}
%\end{IEEEbiography}
%\begin{IEEEbiographynophoto}{Srinivasan Krishnaswamy}
%% insert where needed to balance the two columns on the last page with
%%\newpage


%\vfill


\end{document}